\documentclass[letterpaper, 10 pt, conference]{ieeeconf}  % Comment this line out
\IEEEoverridecommandlockouts                              % This command is only
\usepackage{amsmath,amsfonts,amssymb}
\usepackage{algpseudocode}
\usepackage{algorithm}
\usepackage{array}
\usepackage[caption=false,font=normalsize,labelfont=sf,textfont=sf]{subfig}
\usepackage{textcomp}
\usepackage{stfloats}
\usepackage{url}
\usepackage{verbatim}
\usepackage{graphics}
\usepackage{nccmath}
\usepackage{multirow}
\usepackage{mathtools}
\usepackage{stmaryrd}
\usepackage{xcolor}
\usepackage[english]{babel}
\newtheorem{theorem}{Theorem}

\newtheorem{lemma}{Lemma}

\newtheorem{definition}{Definition}
\newtheorem{assumption}{Assumption}

\usepackage{cite}
\usepackage{hyperref}

\def\QEDclosed{\mbox{\rule[0pt]{1.3ex}{1.3ex}}} % for a filled box
\def\QED{\QEDclosed} % default to closed

\definecolor{color_rich}{RGB}{234, 17, 12}

\usepackage{hyperref}
\hypersetup{
    colorlinks=false,
    linkcolor=blue,
    filecolor=magenta,      
    urlcolor=cyan,
    pdfpagemode=FullScreen,
    }
\usepackage{amsmath}
\usepackage{dsfont}

\newcommand{\R}{\mathbb{R}}

\newcommand{\C}{\mathbb{C}}

\renewcommand{\Re}[1]{\operatorname{Re}\left\{#1\right\}}
\renewcommand{\Im}[1]{\operatorname{Im}\left\{#1\right\}}
\newcommand{\conj}[1]{\mkern 1.5mu\underline{\mkern-1.5mu#1\mkern-1.5mu}\mkern 1.5mu}

\renewcommand{\P}{\operatorname{Pr}}
\newcommand{\E}{\operatorname{E}}

\renewcommand{\j}{j}

\newcommand{\<}{\langle}
\renewcommand{\>}{\rangle}
\newcommand{\T}{\mathsf{T}}

\newcommand{\trace}[1]{\operatorname{trc}\left(#1 \right)}

\newcommand{\ip}[2]{\left\<#1, #2\right\>}
\newcommand{\norm}[1]{\left|\left|#1\right|\right|}
\newcommand{\Expec}[2][]{\E_{#1}\left\{#2\right\}}
\newcommand{\Prob}[2][]{\P_{#1}\left\{#2\right\}}

\newcommand{\p}[1]{\left(#1\right)}
\newcommand{\s}[1]{\left[#1\right]}

\renewcommand{\c}[1]{\left\{#1\right\}}
\newcommand{\abs}[1]{\left|#1\right|}

\newcommand{\set}[1]{\mathcal{#1}}
\DeclareMathOperator*{\argmin}{\operatorname{arg~min}}
\DeclareMathOperator*{\argmax}{\operatorname{arg~max}}

\newcommand{\vone}{\mathds{1}}

\newcommand{\mId}{\mathbf{I}}

\newcommand{\setA}{\set{A}}

\newcommand{\setD}{\set{D}}
\newcommand{\setE}{\set{E}}
\newcommand{\setF}{\set{F}}
\newcommand{\setG}{\set{G}}

\newcommand{\setN}{\set{N}}
\newcommand{\setO}{\set{O}}

\newcommand{\setS}{\set{S}}

\newcommand{\st}{\operatorname{s.\!t.}}

\newcommand{\sgn}{\operatorname{sgn}}
\newcommand{\diag}{\operatorname{diag}}

\renewcommand{\j}{\mathrm{j}}%{\mathbf{j}}
\newcommand{\pmax}{\overline{p}}
\newcommand{\pmin}{\underline{p}}

\renewcommand{\j}{\mathrm{j}}%{\mathbf{j}}

\DeclareSymbolFont{sfoperators}{OT1}{cmss}{m}{n}
\DeclareSymbolFontAlphabet{\mathsf}{sfoperators}

\makeatletter
\def\operator@font{\mathgroup\symsfoperators}
\makeatother
\title{\LARGE \bf
Strategic Electric Distribution Network Sensing via Spectral Bandits
}

\author{Samuel Talkington\textsuperscript{\textsection}, Rahul Gupta\textsuperscript{\textsection}, Richard Asiamah\textsuperscript{\textsection}, Paprapee Buason$^{\dagger}$, and Daniel K. Molzahn\textsuperscript{\textsection}
\thanks{\noindent \textsection: School of Electrical and Computer Engineering,
        Georgia Tech, Atlanta, GA, USA.
        \{talkington, rgupta460, rasiamah3, molzahn\}@gatech.edu. Support from U.S. NSF GRFP DGE-1650044 and PSERC project T-67.}
\thanks{\noindent $^{\dagger}$: Los Alamos National Laboratory. buason@lanl.gov. He is supported by the U.S. DOE Office of Electricity Advanced Grid Modeling program.}%
\vspace{-1em}
}

\begin{document}

\maketitle
\thispagestyle{empty}
\pagestyle{empty}

%%%%%%%%%%%%%%%%%%%%%%%%%%%%%%%%%%%%%%%%%%%%%%%%%%%%%%%%%%%%%%%%%%%%%%%%%%%%%%%%
\begin{abstract}
% %Sensing infrastructure is diffusely deployed in electric distribution networks. Moreover,
% Communication bandwidth is often limited for retrieving real-time information from sensors in electric distribution networks; this poses challenges to their use in monitoring network constraint violations. To solve these challenges, we propose an online bandwidth-constrained sensor sampling algorithm that exploits the graphical structure inherent in the power flow equations. The key idea is to use a spectral bandit framework, where the parameters we wish to estimate are the graph Fourier transform coefficients of the nodal voltages. The structure provided by this framework promotes a sampling policy that strategically accounts for electrical distance.  %, allowing for recursive prediction of nodal criticality from a small subset of sensing devices.
% Maxima of sub-Gaussian random variables model the policy rewards, which relaxes distributional assumptions common in prior work. %Simulations on realistic network models demonstrate the effectiveness of the method.
% The scheme is implemented on realistic electrical networks to dynamically identify meters exposing violations of voltage magnitude limits and illustrating the effectiveness of the proposed method.

Despite their wide-scale deployment and ability to make accurate high-frequency voltage measurements, communication network limitations have largely precluded the use of smart meters for real-time monitoring purposes in electric distribution systems. Although smart meter communication networks have limited bandwidth available per meter, they also have the ability to dedicate higher bandwidth to varying subsets of meters. Using this capability to enable real-time monitoring from smart meters, this paper proposes an online bandwidth-constrained sensor sampling algorithm that takes advantage of the graphical structure inherent in the power flow equations. The key idea is to use a spectral bandit framework where the estimated parameters are the graph Fourier transform coefficients of the nodal voltages. The structure provided by this framework promotes a sampling policy that strategically accounts for electrical distance. Maxima of sub-Gaussian random variables model the policy rewards, which relaxes distributional assumptions common in prior work. The scheme is implemented on a synthetic electrical network to dynamically identify meters exposing violations of voltage magnitude limits, illustrating the effectiveness of the proposed method.
  
\end{abstract}

%%%%%%%%%%%%%%%%%%%%%%%%%%%%%%%%%%%%%%%%%%%%%%%%%%%%%%%%%%%%%%%%%%%%%%%%%%%%%%%%
\section{Introduction}
Sensing infrastructure plays a critical role in the control and operation of electricity networks and is crucial for applications such as state estimation, system identification, and data-driven control schemes~\cite{PMU_placement_for_se, pmu_placement, pmu_placement_2, pmu_placement_review}. For transmission networks, observability is maintained by using phasor measurement units and micro-phasor measurement units spread across multiple substations. This allows for continuous monitoring of the entire system. The same level of real-time observability is not available for distribution networks, however. Smart meters with advanced metering infrastructure (AMI) could help address this problem. 

%For transmission networks, observability is maintained by using phasor measurement units and micro-phasor measurement units optimally placed across the network. The same level of real-time observability is not available for distribution networks, however. Smart meters with advanced metering infrastructure (AMI) could help address this problem. 

The number of smart meters installed throughout the United States continues to increase, reaching 107 million units in 2021~\cite{smart_meter_statistics}, accounting for more than 50\% of all electricity meters installed throughout the country. By regularly measuring and recording nodal voltages, currents, and power injections~\cite{Hubbel}, smart meters have the capability to increase visibility into distribution networks. The frequency and consistency of the recorded data make them good candidates for real-time control and operation applications in distribution networks. Smart meters have the added benefit of being widely available in existing systems, eliminating the need for new infrastructure investments to increase visibility.

Smart meters send measurements to a central server that performs various calculations, often for customer billing purposes. Smart meters communicate these measurements through channels with limited bandwidth connections to the central server~\cite{smart_meter_pinging}. Although the bandwidth of the communication network is sufficient to send power consumption data at 15-minute to hourly intervals for billing purposes, bandwidth limitations pose a significant challenge for real-time monitoring purposes, especially when the limited bandwidth is shared by many smart meters. However, new smart meters can be dynamically queried so that different subsets of the meters can report data at high frequencies. This motivates the development of new algorithms to take advantage of this capability for real-time monitoring purposes.

% This introduces a bottleneck in receiving valuable and timely information from the meters, impacting the situational awareness of the power network. In this context, algorithms need to strategically sample the smart meters, given the restriction on the communication bandwidth, while satisfying specific power system objectives.

This paper focuses on the task of dynamically identifying smart meters to query with the goal of revealing violations of voltage magnitude limits so that system operators can undertake corrective actions. 
%System operators must identify the nodes where substantial power injection fluctuations contribute to grid violations. 
With rapidly growing deployments of distributed energy resources that contribute to substantial variations in power injections, system operators would benefit from algorithms for real-time \emph{online} identification of voltage violations. This requires dynamically selecting varying sets of smart meters to identify voltage violations.
% Traditionally, this problem has been addressed by strategically placing the sensors at specific nodes in the network \cite{}, enabling the inference of violations throughout the entire system. However, when sensors are already deployed, the critical issue lies in optimally sampling those meters to detect grid violations amid alterations in grid injections. Such a strategy necessitates real-time \emph{online} sampling, ensuring that in the event of grid condition alterations, a distinct optimal set of meters is selected for monitoring.

\subsection{Related work}
Smart meter data has been used for real-time monitoring in electric distribution networks \cite{baran_real_time_1994,Angioni2016}, and many other applications in load forecasting, demand response, and consumer characterization applications~\cite{smart_meter_analytics}. Moreover, there is a rapidly growing number of works that develop sensor selection methods for various monitoring tasks in infrastructure network settings. The work of \cite{jezdimir-dahan-network-monitoring-2019} developed a game-theoretic approach to network monitoring, while \cite{zhang_bandit_2023} developed a change-point method based on bandit algorithms\textemdash a highly similar line of research to ours. Distinct from the bandit formulation, \cite{xian_adaptive_2023} developed a method to detect change points. 
More broadly, adaptive sampling techniques have seen recent innovation in continuous environments \cite{grant_adaptive_continuous_2019}. Recent work %similar to ours 
developed an Upper Confidence Bound (UCB)-type algorithm where the rewards are extreme quantities derived from a Gaussian process reward function~\cite{yang_output-weighted_2022}.

Additional authors have developed a wealth of methods for sensor placement in generic network settings. A rich literature from the authors of \cite{jezdimir-dahan-network-monitoring-2019, bahamondes_dahan_imperfect_2022, milosevic_dahan_heterogeneous_2023} has developed analytical zero-sum game approaches to time-invariant network inspection. Additional online efforts have explored reinforcement learning algorithms \cite{abramenko_graph_sampling_reinforcement_2019}.

In the setting of electric distribution networks, sensor selection strategies for network topology identification \cite{cavraro_real-time_2020} have been developed with identifiability guarantees. The work in \cite{krause_efficient_2008} developed submodular algorithms for sensing applications in water distribution networks. Nevertheless, online algorithms of this kind have yet to be utilized for strategic detection of network constraint violations.

% However, none of these works. This would only be made possible via an optimal selection of the meters, which would act as sensors that would provide the most relevant data for enhanced grid observability.

\subsection{Proposed Framework and Contributions}
This work proposes an online resource-constrained sensor sampling algorithm that exploits the graphical structure inherent in the power flow equations. The key idea is to use a spectral bandit framework, where the parameters we estimate are the graph Fourier transform coefficients of the nodal voltages. The structure provided by this framework promotes a sampling policy that accounts for electrical distances.

This paper is part of an extensive and growing literature on multi-armed bandit algorithms \cite{lattimore_bandit_2020,zeng_online_2016}. We are particularly close to works on context-varying linear bandits~\cite{zeng_online_2016}, the large body of work on bandit algorithms leveraging graph spectral structure \cite{valko_spectral_bandits_2014, kocak_spectral_2014, kocak_spectral_2020, thaker_maximizing_2022}, and combinatorial bandits \cite{yue_linear_2011,cesa2012combinatorial}. Recent work in \cite{thaker_maximizing_2022} introduced a graph-based upper confidence bound algorithm similar to ours.

Our work uses the language of bandits, which is a particular sequential game formulation. This enables us to recursively solve the sensor sampling problem to detect violations of constraints on nodal states in a streaming fashion.  To the best of our knowledge, this is the first such work to do this. 

In summary, the contributions of this research are:
\begin{enumerate}
    \item A stochastic LinDistFlow model, developed in Section~\ref{sec:stochastic-lindistflow}, which is agnostic to the probability distributions of the nodal power injections. 
    \item A sufficient condition for nodal voltage magnitudes to be sub-Gaussian random variables; consequently, this yields concentration inequalities for their maximal fluctuations. In particular, these inequalities explicitly depend on the network model topology and parameters.
    \item A new spectral bandit algorithm, developed in Section~\ref{sec:sensor-sampling-algorithm}, which strategically samples a small subset of nodes to find maximal voltage magnitude fluctuations.   
\end{enumerate}

\subsection{Notation}
For a matrix $A \in \C^{n \times d}$, we denote its transpose as  $A^\T$. %We denote the operator norm, or spectral norm, of $A$ as $\opnorm{A} := \sup_{x : \norm{x}_2 \leq 1} \norm{A x}_2$, where $\opnorm{A} = \sigmamax(A) = \sqrt{\lambdamax(A^\H A)}$; here, $\norm{\cdot}_2$ denotes the standard Euclidean norm, $\sigmamax(\cdot)$ denotes the largest singular value, and $\lambdamax(\cdot)$ the largest eigenvalue. 
The $p$ norm of a vector $x \in \R^n$ is denoted as $\norm{x}_p$. The norm of a vector $x \in \R^n$ with respect to matrix $A \in \R^{n \times n}$ is denoted as $\norm{x}_A := \sqrt{x^\T A x}$.
%The symbols $\vone_d,\vzero_d$ are reserved for $d$-dimensional vectors of all ones and zeros, respectively.
The imaginary unit is $\j := \sqrt{-1}$, and $\conj{x}$ is the conjugate of any complex number $x \in \C$. Expectation and probability are denoted as $\Expec[]{\cdot}$ and $\Prob[]{\cdot}$, respectively. We write $a \lesssim b$ if $a \leq Cb$ for some universal constant $C$. The symbol $\vone$ is a vector of all ones, and $\diag(\cdot)$ is a diagonal matrix with entries given by the argument. 

%%%%%%%%%%%%%%%%%%%%%%%%%%%%%%%%%%%%%%%%%%%%%%%%%%%%%
\section{Problem setting}
\label{sec:problem-formulation}
 
% In this paper, we consider an undirected graph $\setG$ corresponding to a distribution network with nodes $\setN := \{1,\dots,n\}$, and lines $\setE \subseteq \setN \times \setN$. 

In this paper, we focus on analyzing an undirected graph $\setG$ that represents a distribution network in power systems, with nodes denoted as $\setN := \c{1,\dots,n}$ and lines represented by $\setE \subseteq \setN \times \setN$. The structure of the graph is akin to that of a power distribution network, although the framework we develop could be extended to other network types.
Let $A \in \{-1,0,1\}^{m \times n}$ be the edge-to-node incidence matrix of the network, and let $w := \s{w_{ij}}_{(i,j) \in \setE}$ be a vector of \textit{edge weights} ordered corresponding to the rows of $A$; they both may model self-edge weights. 
%
%We construct the weighted graph Laplacian matrix as $Y := A^\T \diag(w) A$. 
% Let $\{\lambda_k,q_k\}_{k=1}^n$ be the eigenvalues and eigenvectors of $L$ ordered with $0 \leq \lambda_1 \leq \dots \leq \lambda_n$. Assume that the Laplacian is diagonalizable as $L:= Q \Lambda Q^\T$, where $Q$ is an \emph{orthogonal} matrix, i.e., $Q^\T Q = \mId$, where each column is an eigenvector, and $\Lambda := \diag(\lambda_1,\dots,\lambda_n)$.
%
% Denote the $n$-dimensional vector $x$ as the \emph{nodal state} of the network\textemdash, e.g., nodal voltages, water fluid pressure, or gas pressure, and denote the $n$-dimensional \emph{nodal flow} of the network as $y \in \R^n$. We assume that the Laplacian linearly relates the nodal flows and states as $y := L x = W \Lambda W^\T x$. 
%\vspace{-0.5em}
\subsection{Power flow model for distribution systems}
\label{sec:distribution-network-model}
%A convenient graphical model for a single-phase distribution network can be constructed from mild assumptions. 
We briefly review the graphical model for a single-phase distribution network. The network admittance matrix $Y := A^\T \diag(w) A \in \C^{n \times n}$ is a weighted graph Laplacian matrix that encodes Kirchhoff's current and voltage laws. %We define the complex nodal power injections as $s \in \C^n$, which take the form $s = \diag(x)\conj{Lx}$ with real (active) and imaginary (reactive) components given as $p := \Re{s}$ and $q := \Im{s}$, respectively. 
Let complex nodal voltages be $u \in \C^n$ and nodal power injections be $s \in \C^n$; the two are related by the \emph{injection model} of the power flow equations, which take the form $s = \diag(u)\conj{Y u }$ with real (active) and imaginary (reactive) components given as $p := \Re{s}$ and $q := \Im{s}$, respectively.

\subsection{Review of the LinDistFlow approximation}
We consider a linear power flow approximation known as LinDistFlow to represent the nodal voltage magnitudes in radial distribution networks based on Laplacian matrices; we refer the reader to \cite[Sec. 2]{deka_learning_2024} for additional exposition. Due to the tree structure of distribution networks, the incidence matrix with the first column removed, $A \in \c{-1,0,1}^{n \times n}$, is \emph{square}. Using $r,x \in \R^{m}$ to denote the line resistances and reactances, we can approximate the nodal voltage magnitudes $v \in \R^n$ as $v := v^\bullet \vone + R p + Xq$, where $v^\bullet$ denotes the nominal slack voltage, $R := G^{-1}$ and $X := B^{-1}$, where $G,B\in \R^{n \times n}$ are $n \times n$ Laplacian matrices of the form $G := A^\T \diag(r)^{-1} A$ and $B:= A^\T \diag(x)^{-1} A$.

% Since the Laplacian matrices $G$ and $B$ are symmetric positive-definite, $G,B \succ 0 $ \cite{deka_learning_2024}, they have orthogonal eigendecompositions of the form $G:=W_g \Lambda_g W_g^\T$,  and $B:= W_b \Lambda_b W_b^\T$, where $W_g^\T W_g = W_b^\T W_b = I_n$.  The parameters we wish to estimate are the time-varying \emph{graph Fourier coefficients} of the voltage magnitudes in its appropriate basis, which we denote as $\psi := W^\T v$. 
\begin{assumption}%[Constant power factor]
\label{assum:known-phase}
    % The phase of the injections, $\arctan\p{q_i/p_i} := \arctan\p{\kappa}$, is the same for all nodes $i \in \setN$ in the network.
    The ratio of reactive to active injections, $\kappa = q_i / p_i$, is the same for all nodes $i \in \setN$ in the network.
\end{assumption}
Physically, Assumption \ref{assum:known-phase} is equivalent to the statement that \textit{all power factors} are the same at every node in the network. Importantly, while Assumption \ref{assum:known-phase} aids in the forthcoming theoretical analyses, our numerical results in Section \ref{sec:numerical} empirically indicate that it may be unnecessary in practice. The algorithm appears to perform well in our experiments, even if Assumption \ref{assum:known-phase} is not satisfied.%\footnote{In distribution networks, as in this paper, power factors for inverters are a quantity that can be chosen by the user via remotely programmable settings. Loads have uncontrollable power factors, typically near unity.} 
\begin{lemma}
\label{lemma:system-laplacian}
    Given Assumption \ref{assum:known-phase}, there exists a Laplacian matrix $L = (R + X \mId \kappa)^{-1}$ and orthonormal eigenbasis $W$ such that $L := W \Lambda W^\T$, $\Lambda := \diag(\lambda_1,\dots,\lambda_n)$, and the  LinDistFlow model is then $v = v^{\bullet} \vone + L^\dagger p$ for any $p \in \R^n$, where $\c{\cdot}^{\dagger}$ denotes the Moore-Penrose pseudoinverse, and $v^\bullet$ is a nominal voltage, e.g., 1 per unit (pu).
\end{lemma}
We prove the Lemma in Appendix~\ref{apen:lemma1}. The Laplacian eigenbasis $W$ allows us to construct the \emph{graph Fourier transform} of the nodal injections and voltages, which we write as $\rho := W^\T p$ and $\psi := W^\T v$, respectively. Substituting these transformations into the LinDistFlow model yields the following representation of the voltages
\begin{equation}
   v - v^\bullet \vone =  L^\dagger p = W \Lambda^{-1} W^\T p =  W \psi
\end{equation}
in the \emph{graph Fourier basis} $W$. In the following, the graph Fourier transform coefficients $\psi$ are the parameters we wish to estimate.

\subsection{Action set and reward formulation}
The proposed sensor sampling problem can be expressed naturally in the language of multi-armed bandits. At each measurement interval, or \emph{round} $t$, the learner selects a subset of no more than $b$ nodes to sample to identify possible voltage violations in the network. Thus, the set of all available sampling strategies is defined via the action set 
\begin{equation}
\label{eq:action-set}
    \setA = \c{ \setS \in 2^{\setN} \: : \: \abs{\setS} \leq b},
\end{equation}
of which there are $\abs{\setA} = {n \choose b}$ possible strategies. The action set is a subset of the $n$-dimensional hypercube: $\setA \subset \{0,1\}^n$. This is a classic example of a \emph{combinatorial bandit} problem; see \cite{cesa2012combinatorial,yue_linear_2011}, and \cite[Ch. 30]{lattimore_bandit_2020} for detailed discussions.

\begin{figure}
    \centering
    \includegraphics[width=0.9\linewidth,keepaspectratio]{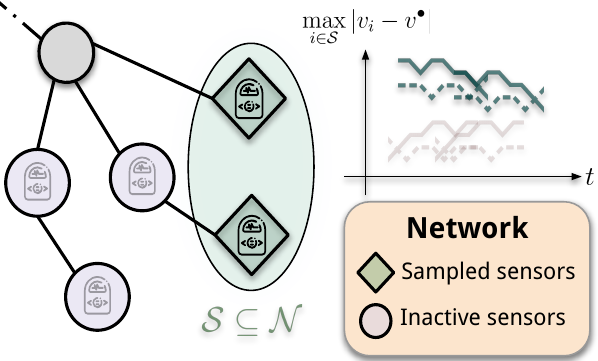}
    \caption{Illustration of the problem: Given an aggregate bandwidth limit across nodes, adaptively design sampling a policy $\setS\subseteq \setN$ to expose violations of voltage magnitude limits.}
    \label{fig:high-level-illustration}
\end{figure}

Often, engineers wish to use sensors in distribution systems to determine how far away nodal voltages are from their nominal values. %this may serve as a proxy for how close a node is to constraint boundaries. 
Thus, an appropriate model for the effectiveness of a sampling strategy $\setS \in \setA$ is the \emph{maximum deviation} from the nominal value observed at any node $i \in \setS$. Concretely, after selecting sensors $\setS \in \setA$ to \emph{query}, the learner then observes a \textit{reward} $f : \setA \to \R$ in the form
\begin{equation}
    \label{eq:reward}
    f(\setS) = \max_{i \in \setS} \abs{v_{i} - v_i^\bullet} = \max_{i \in \setS} \abs{\ip{w_i}{\psi}},
\end{equation}
where $v_i$ and $v_i^\bullet$ are the sampled and nominal voltage at node $i \in \setS \subseteq \setN$, respectively. The voltages $v_i$ are random variables distributed as $v_i \sim \setD_i$, where $\setD_i$ is an arbitrary distribution for all nodes $i \in \setN$. Hereafter, the nominal voltage $v_i^\bullet$ will take the form of the expected value of that nodal quantity, conditioned on past observations, $v_i^\bullet = \Expec[]{v_i \vert \setF_{t-1}}$, where $\setF_{t-1}$ is a filtration, which is a collection of all past input-output pairs $(w_i,v_i)$ up to time $t$.

Note that the reward \eqref{eq:reward} is higher for sampling strategies that expose \emph{extreme deviations} of the nodal quantities from their nominal values. In addition, it has additional appealing properties, described in Lemma \ref{lemma:reward-properties}.
\begin{lemma}
    \label{lemma:reward-properties}
    The reward function $f : \setA \to \R$ is monotone submodular in $\setA$ and $1$-Lipschitz (see the proof in Appendix~\ref{apen:lemma2}).
\end{lemma}
% We prove Lemma \ref{lemma:reward-properties} in the Appendix~\ref{apen:lemma2}.

% The learning setting is as follows. Each node $i \in \setN$ possess a \emph{nodal state} $x_i \in [-1,1]$ and the undirected graph has weights $\c{w_{i,k}}_{(i,k) \in \setE}$ where $w_{i,k} \in [0,\infty)$ is the magnitude of the link between nodes $(i,k)$.

At each \emph{round} or time step, $t=1,\ldots,m$, the power injections $p_1,\ldots,p_t \in \R^n$ are a stochastic process. The learner then selects subsets of sensors $\setS_1,\ldots,\setS_t \subseteq \setN$ to \emph{query}. The \emph{regret} of the learner over $m$ time steps is
\begin{equation}
    \label{eq:regret}
    R_m := \Expec[]{\max_{\setS \in \setA} \sum_{t=1}^m f(\setS) - f(\setS_{t})},
\end{equation}
which the learner wishes to minimize; we, in turn, wish to bound it. The regret \eqref{eq:regret} models the cumulative differences between the maximal average perturbation around a nominal point and the maximal perturbations observed by our chosen sampling policy.

%%%%%%%%%%%%%%%%%%%%%%%%%%%%%%%%%%%%%%%%%%%%%%%%%%%%%
\section{Stochastic LinDistFlow model}
\label{sec:stochastic-lindistflow}
This section proposes a stochastic, distribution-agnostic LinDistFlow model. Given a mild assumption on power consumption in distribution grids, we use the structure provided by LinDistFlow to probabilistically model the extrema of voltage magnitude fluctuations about a nominal point.
% \subsection{Concentration of bounded sub-Gaussian sequences}
% \label{sec:main-results-subgaussian}
Many of our results are based on sub-Gaussian concentration. We recall the definition of these random variables below.
\begin{definition}
    \label{def:sub-gaussian}
    A random variable $x \in \R$ is sub-Gaussian with parameter $\sigma$ if $\Expec[]{e^{\lambda \p{x - \Expec[]{x}}}} \leq e^{\lambda^2 \sigma^2/2}$ for all $\lambda \in \R$. %where $\Expec[]{e^{\lambda(x - \Expec[]{x})}}$ is the \emph{moment-generating function} of the random variable $x$.
\end{definition}
%For each node $i \in \setN$, let  $x_i$ be the nodal state and let $\Expec[]{x_i} := x^\bullet_i$ be its nominal value. 
We first make a general statement on the fluctuation of random nodal voltages in the LinDistFlow model around their nominal values. Critically, this statement holds for a broad family of possible active power injection distributions\textemdash that is, any distribution with bounded support. Specific knowledge of power injection distributions is unnecessary.

We emphasize that this work does not assume specific distributions for any random quantity. Instead, a \emph{family} of possible distributions is considered. Concretely, for each sampling interval $t= 1,\ldots,m$, let $p_t := p_\star + \tilde{p}_t \in \R^n$  be the nodal injections, where $p_\star$ is a nominal value, and $\tilde{p}_t$ is a random change in the nodal injections drawn from an unknown distribution. The only assumption is that the \textit{range} of power consumption can be inferred, as we describe now.
\begin{assumption}
    \label{assum:reporting-scenarios}
    Assume that the power consumption at each node $i \in \setN$ is an independent variable whose range is a \textit{predictable process}, and therefore we can infer bounds on the injections $\{(\pmin_t,\pmax_t)\}_{t=1}^m$ from historical data such that $ \pmin_t \leq \tilde{p}_t \leq \pmax_t$ almost surely. 
\end{assumption}
Hereafter, denote $\Delta_t := \pmax_t - \pmin_t >0$. 

\subsection{Concentration of nodal voltage perturbations}

Using the structure of the LinDistFlow model, Assumption~\ref{assum:reporting-scenarios} immediately gives rise to concentration bounds for the voltages.  Moreover,  Assumption~\ref{assum:reporting-scenarios} relaxes the Gaussianity assumptions common in prior work (e.g. \cite{li2011phasor, kekatos2012optimal}), which we discuss in Appendix~\ref{apdx:proof:subgaussian-vmag}.  The following results are valid even if $(\pmin_t,\pmax_t)$ are the largest physically plausible injections, i.e., grid constraint bounds. Access to tighter bounds forecast by the sensors only serves to make the results more precise.
\begin{lemma}
\label{lemma:subgaussian-vmag}
    % Let $\setF_t = \sigma(w_1,v_1,\dots,w_{t-1},v_{t-1},w_t) $ be the $\sigma$-algebra containing information just before the voltage $v_t$ is observed, and let $\Expec[t]{\cdot} := \Expec[]{\cdot \vert \setF_{t}}$.
     If Assumption \ref{assum:reporting-scenarios} holds, then for each $i \in \setN$ we have that $v_i - v^\bullet$ is a sub-Gaussian random variable with parameter $\frac{1}{2}\Delta \norm{\Lambda^{-1} w_i}_2$.
\end{lemma}
Lemma \ref{lemma:subgaussian-vmag}, proven in Appendix \ref{apdx:proof:subgaussian-vmag}, characterizes a family of possible distributions on the voltage magnitudes that arise from Assumption \ref{assum:reporting-scenarios}. Applying sub-Gaussian maxima concentration leads to our first primary result.

% \begin{assumption}
% \label{assum:bounded-voltages}
%     Fix $\xmin < \xmax$ as worst-case minimum and maximum states, respectively, such that we may assume $\xmin\leq x_i \leq \xmax$ for all $i \in \setN$, almost surely. 
% \end{assumption}
%  While Assumption \ref{assum:bounded-voltages} is very mild, it still allows us to develop guarantees on the likelihood of observing such violations from a given sampling strategy. A principled choice for $\xmin$ and $\xmax$ are the network constraints $\pm$ tolerance $\epsilon>0$. We immediately obtain the following result.
\begin{theorem}
\label{thm:max-v-concentration}
     Let $\setS \subseteq \setN$ be a sampling of $b$ nodes. Suppose that $\Delta_t := \Delta$ for all $t$, and suppose that LinDistFlow accurately represents the network model. If Assumption \ref{assum:reporting-scenarios} holds, we have % then for each $i \in \setN$ we have that $v_i - v^\bullet$ is a sub-Gaussian random variable with parameter $\frac{1}{2}\Delta \norm{\Lambda^{-1} w_i}_2$, and
     %$\frac{1}{2}\Delta$ for all $i \in \setS$, and
    \begin{equation}
    \label{eq:l1:bounded-x-expec}
        \Expec[]{\max_{i \in \setS} \abs{v_i - v_i^\bullet}} \lesssim \frac{1}{2}\Delta \max_{i \in \setS} \norm{\Lambda^{-1} w_i}_2^2\sqrt{2\log(b)};
    \end{equation}
    moreover, for all $\epsilon > 0$
    \begin{equation}
    \label{eq:l1:bounded-x-prob}
        \Prob[]{\max_{i \in \setS} \abs{v_i -  v_i^\bullet} > \epsilon} \leq 2b \exp\c{ \frac{-2 \epsilon^2}{\Delta^2 \max\limits_{i \in \setS} \norm{\Lambda^{-1} w_i}_2^2 }  }.%\leq 2d \exp\left( -\frac{2\epsilon^2}{\Delta^2}\right).
    \end{equation}
\end{theorem}
Theorem \ref{thm:max-v-concentration} is proven in Appendix \ref{apdx:proof:max-x-concentration}. There is an intuitive physical interpretation for these results. The results provide probabilistic predictions of the worst-case fluctuation of nodal voltages that would be observed by a sensor sampling strategy $\setS \subseteq \setN$; these predictions can inform how sensors should be selected. The result is achieved by virtue of the graphical structure of the LinDistFlow model, coupled with the mild Assumption \ref{assum:reporting-scenarios} on bounded load behavior.

\subsection{Condition for positive semi-definite Laplacian}
\label{sec:psd}

Regret guarantees for spectral bandits typically assume that the Laplacian $L$ is positive semidefinite. This is not guaranteed in general for distribution networks\textemdash the power factor ratio $\kappa$ or the reactances $x$ may be negative, potentially rendering $L$ indefinite \cite{zelazo_laplacian_definite_2014}. However, we can construct a condition using the physical structure of distribution networks to ensure that $L$ satisfies this requirement. This pursuit begins with an assumption that mirrors \cite{chen_negative_weights_2016}.
\begin{assumption}
    \label{assum:effective-resistance-bound}
    Let $\kappa \in \R$ as in Assumption \ref{assum:known-phase}, and let $\setE_+$ (resp. $\setE_- $) be the subsets of all lines $(i,j) \in \setE$ where $r_{ij} + \kappa x_{ij}$ is positive (resp. negative). Let $L_+$ be the Laplacian constructed from the subnetwork $(\setN,\setE_+)$. Assume that
    \begin{equation}
        \abs{r_{ij} + \kappa x_{ij}} \geq e_{ij}^\T L^\dagger_+ e_{ij} \quad \forall (i,j) \in \setE_-,
    \end{equation}
    where $e_{ij} := e_i - e_j$ is the difference between the $i$-th and $j$-th basis vectors and  $e_{ij}^\T L^\dagger_+ e_{ij}$ is the effective resistance between nodes $i,j \in \setN$.
\end{assumption}
Intuitively, Assumption \ref{assum:effective-resistance-bound} requires the effective impedance between any two points in the network to be non-negative. The above condition leads to the following result.
\begin{lemma}
    \label{lemma:psd-L}
    If and only if Assumption \ref{assum:effective-resistance-bound} holds, 
    \begin{equation}
        L \succeq 0.
    \end{equation}
\end{lemma}
\begin{proof}
    If the network lacks cycles, Assumption \ref{assum:effective-resistance-bound} is necessary and sufficient for the claim to hold; see \cite[As. 3.1]{chen_negative_weights_2016}. By radiality, a single-phase distribution network lacks cycles; the claim is then a corollary of \cite[Thm. 3.2]{chen_negative_weights_2016}.
\end{proof}

\subsection{Confidence ellipsoid, sampling policy, and regret}
\label{sec:regret_analysis}
We now provide a bound on the regret of the strategic sampling policy, when $v_t$ is generates as noisy observations of the form $(v_t)_i - v^\bullet_i = \ip{r_i + \kappa x_i}{p_\star} +\eta_t = \ip{w_i}{\psi_\star} +(\eta_t)_i$, where $(\eta_t)_i = \ip{r_i + \kappa x_i}{\tilde{p}_t}$.
\begin{theorem}
    \label{thm:regret}
    Let $V_t = \Lambda + \sum_{t=1}^m w_s w_s^\T$, then there exists a $\delta \in (0,1)$ and constants $c_t$ for all $t =1,\ldots,m$ such that for any $w \in \R^n$, with probability $1-\delta$:
    \begin{equation}
        % \setC_t = \c{ \psi \in \R^n \: : \: \norm{\hat{\psi}_t - \psi }_{V_{t-1}^2} \leq c_t }.
        \abs{w^\T\p{\hat{\psi}_t - \psi_\star}} \leq c_t \norm{w}_{V_{t-1}^{-1}}.
    \end{equation}
    Moreover, let $b=1$. If Assumptions \ref{assum:known-phase}, \ref{assum:reporting-scenarios}, and \ref{assum:effective-resistance-bound} hold, and $\abs{w_i^\T \psi_\star} \leq 1$ for all $i \in \setN$, the regret \eqref{eq:regret} of the sampler over $m$ periods is bounded as
    \begin{equation}
        R_m \leq \tilde{\setO}(d \sqrt{m}),
    \end{equation}
    where $d$ is the effective dimension, defined in \cite{valko_spectral_bandits_2014} as
    \begin{equation}
        \label{eq:eff-dim}
        d := \max_{i \in \setN}  \ i \quad \st \quad (i-1)\lambda_i \leq \frac{m}{\log(1+m/\lambda_1)}, %\bigg\lceil \frac{\max \log \prod_{i \in \setN}(1 + \frac{t_i}{\lambda_i} }{\log(1+ \frac{m}{k \lambda})} \bigg\rceil,
    \end{equation}
    where $\lambda_1$ is the smallest eigenvalue of $L$.
\end{theorem}
The result in Theorem \ref{thm:regret} can essentially be understood as stating that the regret of the proposed sampling algorithm grows sublinearly, with a scaling factor \eqref{eq:eff-dim} that is at most the number of nodes, $d \leq n$. 

Hereafter, for a sampling strategy $\setS \subseteq \setN$, we let $W_{\setS} \in \R^{b \times n}$ denote the submatrix of the graph Fourier basis $W$ with rows corresponding to $\setS$. It is straightforward to combine \cite[Thm. 1]{yue_linear_2011} by Lemma \ref{lemma:reward-properties}, while using the regularization framework of \cite{valko_spectral_bandits_2014} to produce a $\tilde{\setO}(d\sqrt{bm})$ bound. We leave the completion of this proof to future work.

%%%%%%%%%%%%%%%%%%%%%%%%%%%%%%%%%%%%%%%%
\section{Sensor Sampling Algorithm}
\label{sec:sensor-sampling-algorithm}

The greedy algorithm in \cite{Nemhauser1978AnAO} approximately maximizes a normalized, monotone, submodular objective function $F:2^{\setN} \mapsto \R$ subject to cardinality constraints. In particular, the algorithm solves the program $\max_{\setS \in 2^{\setN}}  F(\setS)$, subject to $\abs{\setS} \leq b$, within a tolerance of $(1 - \frac{1}{e})$ of the global solution.

% , via the procedure below. 
% \begin{algorithmic}
% \State $\setS \gets \varnothing$
% \State $i \gets 1$
% \While{$i \leq b$}
%     \State $k_i \gets \argmax_{j \in \setV} f(\setS_{i-1}\cap\{j\}) - f(\setS_{i-1})$
%     \State $\setS_i \gets \setS_{i-1} \cap \{k_i\}$
%     \State $i \gets i +1$
% \EndWhile
% \end{algorithmic}
% The algorithm is called \emph{greedy} because it selects the index $j$ that provides the highest increase in $F(\cdot)$ and adds it to the set $\setS_i$ for iterations $i \leq b$.

This algorithm will be integrated into our bandit framework; see \cite{yue_linear_2011} for an explanation on how the integration of this step into the algorithm affects the regret bound.

% \subsubsection{Greedy confidence bounds}
% \begin{enumerate}
%     \item 1-node at a time
%     \item Top $b$ nodes
%     \item Top $b$ nodes with time-varying noise
% \end{enumerate}

\subsection{Optimism principle}
To construct our proposed algorithm, we first define an optimistic upper bound on the reward function $\tilde{f} : \setA  \to \R$. This reward leverages the confidence ellipsoid developed in Section \ref{sec:regret_analysis}; it takes the form
\begin{equation}
    \tilde{f}(\setS) = \max_{i \in \setS} \abs{w_i^\T \psi - v^\bullet_i} + c \norm{w_i}_{V_{t}^{-1}},
\end{equation}
where $c$ is the exploration parameter.
\begin{algorithm}[tb]
\caption{\textsc{ExtremeSpectralSampler}}
\label{alg:GroupSpectralUCB}
\begin{algorithmic}[1]
\Require 

$\{n,m,b\}$: \# nodes, pulls, and samples/pull,

$\{\Lambda,W\}$: spectral basis of $L$,

$\{\lambda,\delta\}$: regularization and confidence parameters,

$\{\pmin,\pmax\}$: report bounds on nodal injections,

% $\{A,r,x\}$: network topology parameters,

% $\{\alpha\}$: power factor

\State $\Delta \gets \pmax - \pmin$
\State $d \gets \eqref{eq:eff-dim}$%(\Lambda,\lambda,m)$
\For{$t=1,\dots,T$}
    \State $\setS_0 \gets \varnothing$, 
    \State $i \gets 1$
    \While{$i \leq b$} \Comment{Choose $b$ nodes ($b$ rows of $W$)}
        \State $k_i \gets \argmax_{j \in \setN} \tilde{f}(\setS_{i-1}\cup\{j\}) - \tilde{f}(\setS_{i-1})$
        \State $\setS_i \gets \setS_{i-1} \cup \{k_i\}$
        \State $i \gets i +1$
    \EndWhile
    \State Observe noisy voltages $v_t \gets v^\bullet + W_{\setS} \psi + \eta_t$ 
    \State Observe noisy reward $f_t \gets \max_{i \in \setS_t} \abs{w_i^\T \psi_\star + \eta_t}$
    \State $V_{t+1} \gets  V_t + W_{\setS}^\T W_{\setS} $ 
    \Comment{Update Fourier coefficients}
    \State $\hat{\psi}_{t+1} \gets V_{t+1}^{-1} \sum_{s=1}^t W_{\setS}^\T (v_t-v^\bullet)$ 

\EndFor
\end{algorithmic}
\end{algorithm}
In Algorithm \ref{alg:GroupSpectralUCB}, $W_{\setS} \in \R^{b \times n}$ is the submatrix formed from the $b$ rows of the basis $W$ indexed by $\setS$.

\subsection{Spectral UCB}
Online decision-making problems in which rewards are generated from an underlying graphical model can be solved with the spectral bandit framework \cite{valko_spectral_bandits_2014,kocak_spectral_2020}. These algorithms recursively solve a regularized least squares program
\begin{equation}
\label{eq:spectral-ucb}
    \hat{\psi}_t = \argmin_{\psi \in \R^n} \sum_{s=1}^{t-1} \p{ v_s - \ip{w_{s}}{\psi}}^2 + \beta \norm{\psi}_\Lambda^2,
\end{equation}
where $w_{s} \in \R^n$ is the row of the graph Fourier basis $W$ selected at time step $s$, which in our case, corresponds to having sampled the sensor at the node $i_s \in \setN$. The regularization parameter $\beta >0$ is chosen by the user; a particular choice of $\beta$ that depends on the \emph{effective dimension} of the Laplacian allows the program \eqref{eq:spectral-ucb} to achieve superior regret in comparison to standard least squares \cite[Thm. 8]{kocak_spectral_2020}. The regularizer $\norm{\cdot}_{\Lambda}$ is the following norm defined in terms of the graph Laplacian:
\begin{equation}
\label{eq:spectral-regularization}
    % \norm{\psi}_\Lambda := \sqrt{\psi^\T L \psi} = \sqrt{\sum_{(i,j) \in \setE} w_{ij}(\psi_i - \psi_j)^2}.
    \norm{\psi}_\Lambda := \sqrt{\psi^\T \Lambda \psi} = \sqrt{p^\T W \Lambda^{-1} Wp} = \sqrt{p^\T L^{-1} p}.
\end{equation}
The regularizer \eqref{eq:spectral-regularization} promotes spatially smooth voltage predictions. Specifically, this regularizer promotes solutions where predicted voltages are similar at well-connected nodes; see \cite{valko_spectral_bandits_2014} for further details. The program \eqref{eq:spectral-ucb} has a closed form solution at each time $t$:
\begin{equation}
    \hat{\psi}_t = \p{ \sum_{s=1}^{t-1} w_s w_s^\T + \beta \Lambda  }^{-1} \p{ \sum_{s=1}^{t-1} w_s v_s}.
\end{equation}

%%%%%%%%%%%%%%%%%%%%%%%%%%%%%%%%%%%%%%%%%
\section{Numerical Experiments}
\label{sec:numerical}
% Let $Y \in \C^{n \times n}$ be the nodal admittance matrix and let $v,s \in \C^n$ be the nodal complex voltage and power injections. Let $L:= - \Im{Y} \succeq 0$ be the graph Laplacian, let $x := \operatorname{atan2}(v) \in (-\pi,\pi]^n$ be the nodal state, and let $y := \Re{s}\in \R^n$ be the nodal observations. The measurement system given by the LinDistFlow approximation is $y = Lx$.

Experimentally, we demonstrate the proposed method on the \texttt{case33bw} network \cite{baran_network_1989}. A demonstration of the online sampling algorithm is shown in Fig. \ref{fig:maximal_reward_regret} for the simple case where all power factors are unity ($\alpha_i =1$ for each node $i \in \setN$). The results with the spectral regularizer \eqref{eq:spectral-regularization} in the left subfigures are juxtaposed with the conventional LinUCB algorithm\textemdash which uses $\ell_2$-norm regularization\textemdash shown in the rightmost subfigures. As expected, a sublinear growth in regret is observed for both, and the growth rate is markedly improved for the SpectralUCB algorithm compared to the LinUCB algorithm.
\begin{figure}[!t]
    \centering
    \includegraphics[width=0.98\linewidth,keepaspectratio]{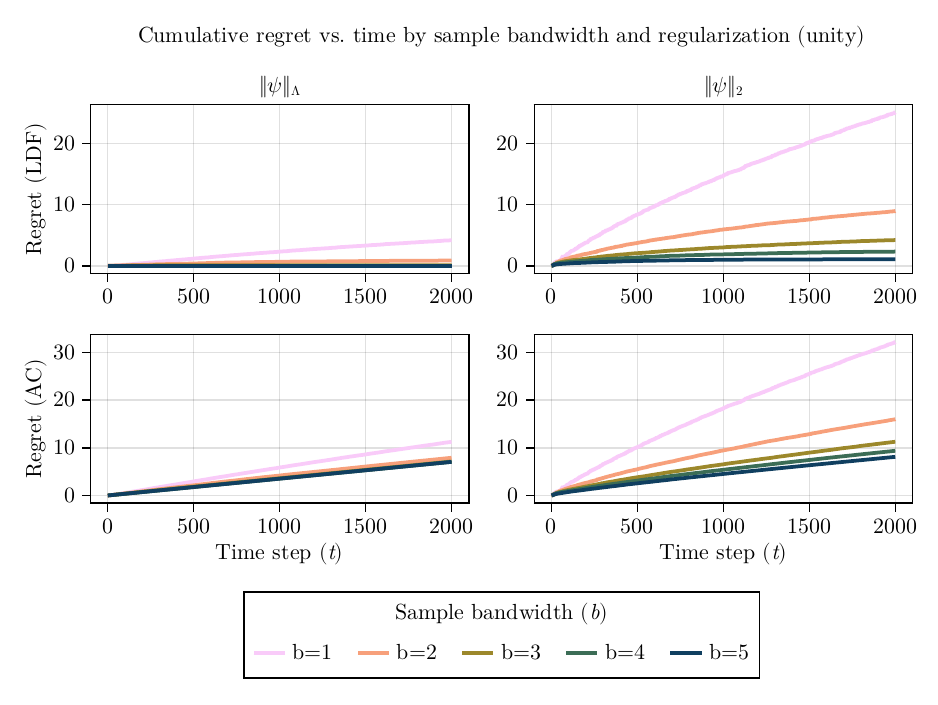}
    % \caption{Regret \eqref{eq:regret} of the bandwidth-constrained maximal voltage risk sampler vs. time with spectral (left) and $\ell_2$ (right) regularization.}
    \caption{Performance with unity power factor at all nodes. The plots show regret relative to the LinDistFlow and AC optimal sampling vs. time with spectral (left) and $\ell_2$ (right) regularization.}
    \label{fig:maximal_reward_regret}
    \vspace{-1em} 
\end{figure}
\begin{figure}[!t]
\vspace{-0.2em}
    \centering
    \includegraphics[width=0.98\linewidth,keepaspectratio]{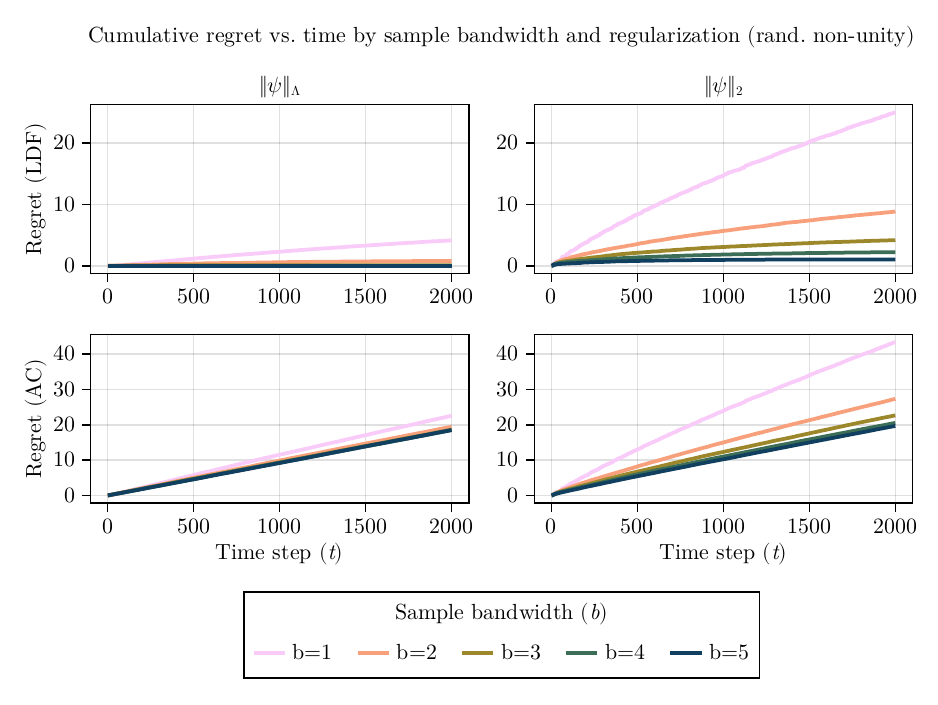}
    \caption{Performance with random, non-unity power factors. The plots show regret relative to the LinDistFlow and AC optimal sampling vs. time with spectral (left) and $\ell_2$ (right) regularization.}
    \label{fig:maximal_reward_regret_nonunity}
    \vspace{-1em}
\end{figure}
% \begin{figure}[!t]
% \vspace{-0.2em}
%     \centering    \includegraphics[width=0.98\linewidth,keepaspectratio]{figures/maximal_reward_regret_by_bandwidth_random_fixed_pf_plus_ac.pdf}
%     \caption{Performance of \textsc{ExtremeSpectralSampler} for random, non-unity power factors: Regret relative to the AC optimal solution vs. time with spectral (left) and $\ell_2$ (right) regularization. }
%     \label{fig:ac-regret}
%     \vspace{-1em}
% \end{figure}
% While the result in Fig. \ref{fig:maximal_reward_regret} aligns with Assumption \ref{assum:known-phase}\textemdash i.e., all nodes are unity power factor\textemdash the results in Fig. \ref{fig:maximal_reward_regret_nonunity} indicate that this Assumption need not be true in practice. Importantly, in the experiment depicted in Fig. \ref{fig:maximal_reward_regret_nonunity}, we do not assume that power factors are uniform
% % , in contrast with Assumption \ref{assum:known-phase}. 
% Instead, the reactive power control parameters vary for each node. We generate these control parameters as random, fixed power factors bounded between 0.90 to 1, i.e., $0 \leq \kappa \leq 0.48$. Similarly, the signs of the reactive power injections are generated according to a Rademacher distribution, i.e., $\sgn(\kappa) = \pm 1$ equiprobably. These parameters allow us to construct a reactive power injection for any fixed active power injection, as we discuss in Appendix \ref{apen:lemma1}.

In the next experiment, we do not assume unity power factors. Instead, the reactive power control parameters vary for each node. We generate these control parameters randomly as fixed power factors bounded between 0.90 and 1, that is, $0 \leq \kappa \leq 0.48$. Similarly, the signs of the reactive power injections are generated according to a Rademacher distribution, i.e., $\sgn(\kappa) = \pm 1$ equiprobably. These parameters allow us to construct reactive power injections for any fixed active power injection, as we discuss in Appendix \ref{apen:lemma1}. Similarly to the finding in Fig.~\ref{fig:maximal_reward_regret}, the result shown in Fig.~\ref{fig:maximal_reward_regret_nonunity} illustrates a significant reduction in cumulative regret using the SpectralUCB algorithm compared to the LinUCB algorithm.

Finally, we compute an ``AC regret" quantity. Specifically, over the entire time horizon  $t=1,\ldots,m$, we precompute voltage magnitudes corresponding to the true AC power flow solution. This differs from the previous experiments depicted in the top subfigures of Figs. \ref{fig:maximal_reward_regret} and \ref{fig:maximal_reward_regret_nonunity}, where the LinDistFlow approximation of the voltage magnitudes serves as the ground truth. Instead, we compute the regret \eqref{eq:regret} as $R_t = \Expec[]{\sum_{\tau=1}^t f(\setS_{\sf ac}) - f(\setS_t)}$, where $\setS_{\sf ac}$ is the optimal clairvoyant sampling strategy given from the AC solutions. The AC regret results shown in the bottom subfigures of Fig.~\ref{fig:maximal_reward_regret} and Fig.~\ref{fig:maximal_reward_regret_nonunity} similarly show a reduction of the growth rate in regret when using the SpectralUCB algorithm.
%%%%%%%%%%%%%%%%%%%%%%%%%%%%%%%%%%%%%%%%%
\section{Discussion and Conclusion}
This work proposed a method to selectively sample streaming sensors in distribution networks according to a security criterion\textemdash the maximal voltage fluctuations around a nominal value. The key idea is to embed the graphical structure of distribution networks into a spectral bandit algorithm, promoting an electrically dispersed sampling policy. 

In contrast with previous work that relies on Gaussian assumptions, we used a distribution-agnostic framework to develop a stochastic LinDistFlow model. This enabled us to develop a principled model for the perturbations of nodal voltages around an operating point without knowledge of the specific distribution driving network loads.

Although the greedy submodular maximization approach we used is well known for solving optimization problems over sets\textemdash including bandit problems\textemdash it has limitations. Chief among these limitations are the limited convergence guarantees and its inherent approximation error in the objective function. Future work by the authors will explore more modern extreme bandit frameworks, such as the inverse gap weighting approach \cite{sen_topk_2021}.

Critically, the strongest theoretical assumption we made to simplify our analysis, Assumption \ref{assum:known-phase}, did not appear to significantly affect the numerical experiments in Section \ref{sec:numerical}. Moreover, while we have presently used synthetic test data, ongoing work by the authors is developing a testbed to investigate the performance of the proposed algorithms on actual networks.

\section*{Acknowledgement}
The authors gratefully acknowledge D. Rikken at Hubbell, Inc. for insightful discussions on communication network capabilities for distribution system sensors. The authors also gratefully acknowledge J. Chen, A. Rangarajan, P. Aquino de Alcantara, L. Roald, F. B. Costa, and D. Fuhrmann for stimulating group conversations.
\bibliographystyle{ieeetr}
\bibliography{refs}

\appendix
\section{Proofs of Lemmas}
\subsection{Proof of Lemma \ref{lemma:system-laplacian}}
\label{apen:lemma1}
\begin{proof}
    First, note that for any fixed active injection $p$ and known reactive settings $\alpha \in (0,1]^n$ and $\sgn(q) \in \p{0,\pm 1}^n$, we can express the corresponding reactive injection $q \in \R^n$ as $q := Kp$, where $K$ is a diagonal matrix with entries that take the form $K_{ii} = \sgn(q_i) \alpha_i^{-1}\sqrt{1-\alpha_i^2}$ for all $i \in \setN$. Then, note that
    \begin{subequations}
        \begin{align}
        v &= v^\bullet \vone + Rp + Xq = v^\bullet \vone + (R + XK)p\\
       % &= v^\bullet \vone + A^{-1} \diag(r) A^{-\T}p + A^{-1} \diag(\kappa x) A^{-\T}  p\\
        &= v^\bullet \vone + A^{-1} \diag(r + \kappa I x) A^{- \T} p.
        \end{align}
    \end{subequations}
    Now let $L^\dagger := R + XK$, where $R,X$ are inverses of Laplacian matrices. Under Assumption \ref{assum:known-phase}, $K_{ii} := \kappa$ for all $i\in\setN$, hence $L^\dagger$ is symmetric, as $L^{-\T} = (R^\T + K^\T X^\T) = (R + \kappa I X) = L^\dagger$. As $L^\dagger$ is real and symmetric, it has an orthonormal eigenbasis $W \in \R^{n \times n}$ such that $L^\dagger := W \Lambda^{-1} W^\T$, where $\Lambda^{-1} = \diag(\lambda_1^{-1},\dots,\lambda_n^{-1})$ is a diagonal matrix of the reciprocals of the Laplacian eigenvalues.
\end{proof}
    % By Hua's identity,
    % \begin{align}
    %     L &= (R+\kappa I X)^{-1} = R^{-1} - (R + \kappa^{-1}R X^{-1} R)^{-1}\\
    %     &= W_g \Lambda_g^{-1} W_g^\T - \p{ W_g \Lambda_g W_g^\T + \kappa^{-1} W_g \Lambda_g W_g^\T W_b \Lambda_b^{-1} W_b^\T W_g \Lambda_g W_g^\T }^{-1}\\
    %     &=W_g \Lambda_g^{-1} W_g^\T - \p{ W_g \p{ \Lambda_g  + \kappa^{-1} \Lambda_g W_g^\T W_b \Lambda_b^{-1} W_b^\T W_g \Lambda_g }  W_g^\T }^{-1}\\
    %       &=W_g \Lambda_g^{-1} W_g^\T - \p{ W_g  \Lambda_g \p{ \Lambda_g^{-1}  + \kappa^{-1} W_g^\T W_b \Lambda_b^{-1} W_b^\T W_g}  \Lambda_g  W_g^\T }^{-1}
    % \end{align}
    
    % We want to show that $(R+XK)^{-1}$ is Laplacian. To show positive definiteness, note that $u^\T L u >0$ for all $u \in \R^n$ if and only if $u^\T L^\dagger u >0$ for all $u \in \R^n$. We obtain that
    % \begin{subequations}
    %     \begin{align}
    %         u^\T L u &:= u^\T R u + u^\T X K u\\
    %         &= \underbrace{u^\T R u}_{\geq 0} + \sum_{i \in \setN} u_i K_{ii} \ip{x_i}{u}
    %     \end{align}
    % \end{subequations}
    % for all $u \in \R^n$. Thus, $L \succeq 0$ and the proof is complete.
% \end{proof}
%
\vspace{-0.5em}
\subsection{Proof of Lemma \ref{lemma:reward-properties}}
\label{apen:lemma2}
\newcommand{\frakS}{\mathfrak{S}}
% \vspace{-1em}
\begin{proof}
    Let $a : \setA \to \{0,1\}^n$ map a sampling strategy $\setS \in \setA$ to a binary vector whose entries take $a_i = 1$ if $i \in \setS$ and $0$ otherwise. Notice that we can equivalently write $f : \setA \to \R$ for any $\setS \in \setA$ as
    \begin{equation}
        f(\setS) = \norm{\diag\p{ a\p{\setS} } L^\dagger (y - y_i^\bullet) }_\infty,
    \end{equation}
    where $\norm{\cdot}_\infty : \R^n \to \R_+$ is the $\ell_\infty$ norm. Now, recall that any norm $\norm{\cdot} : \R^n \to \R_+$ is $L$-Lipschitz, where $L$ is the smallest $L >0$ that satisfies $\norm{x} \leq L \norm{x}_2$ for any $x \in \R^n$. By norm equivalence, we have $\norm{x}_\infty \leq \norm{x}_2$ for any $x \in \R^n$. We conclude that $f$ is $1$-Lipschitz.

    Monotonicity clearly holds; for all $\setS,\frakS \in 2^{\setN}$ such that $\setS \subseteq \frakS$, we have $f(\setS)  \leq f(\frakS)$. Furthermore, normalization also follows clearly as $f(\emptyset) = 0$. 
    
    Let the centered voltages be $\tilde{v}_i := v_i - \Expec[]{v_i}$ for all $i \in \setN$. To establish submodularity, we want to show that for all $\setS \subseteq \mathfrak{S} \in 2^{\setN}$ that
    \begin{equation}
            \label{eq:def-submodularity}
            %\forall \setS \subseteq \setT \in 2^{\setV}, \ \forall s \in \setV \setminus\{\setT\}, \ 
            % f(\setS \cup \{s\}) - f(\setS) \geq f(\frakS \cup \{s\}) - f(\frakS).
            f(\setS) + f(\frakS) \geq f(\setS \cap \frakS) + f(\setS \cup \frakS).
        \end{equation}
        Using the mapping $a : \setA \to \c{0,1}^n$ as before, we obtain
        \begin{equation}
            \begin{split}
                f(\setS \cap \frakS) &+ f(\setS \cup \frakS) = \norm{\p{a(\setS)  \circ a(\frakS) } \circ \tilde{v} }_\infty\\
                &+ \norm{\p{ a(\setS) + a(\frakS) - a(\setS) \circ a(\frakS) } \circ \tilde{v}}_\infty. 
            \end{split}
        \end{equation}
        We complete the proof in cases. First, suppose that $\setS \subseteq \frakS$; this means $a(\setS) \circ a(\frakS) \leq a(\setS)$, where the inequality is elementwise. The inequality is legal because, for all $i \in \frakS \setminus \setS$, we have $a(\setS)_i a(\frakS)_i = 0$. Hence,
        \begin{equation}
        \label{eq:l1:submodular-bound1}
             \begin{split}
             \!\!\!\!f(\setS \cap \frakS) + f(\setS \cup \frakS) &\leq \norm{a(\setS) \circ \tilde{v} }_\infty + \norm{ a(\frakS) \circ \tilde{v}}_\infty\\
             &:= f(\setS) + f(\frakS).
             \end{split}
        \end{equation}
        For the second case, take $\frakS \subseteq \setS$, and note that we can equivalently have $a (\setS) \circ a(\frakS) \leq a(\frakS)$ and thereby produce the same bound \eqref{eq:l1:submodular-bound1}. For the final case where $\frakS \cap \setS = \emptyset$, by applying the triangle inequality, we obtain
        \begin{equation}
        \begin{split}
            f(\setS \cap \frakS) + &f(\setS \cup \frakS) = \norm{\p{ f(\frakS) + f(\setS) } \tilde{v}}_\infty\\
            &\leq \norm{f(\frakS) \circ \tilde{v}}_\infty + \norm{f(\setS) \circ \tilde{v}}_\infty\\
            &:= f(\setS) + f(\frakS).
        \end{split}
        \end{equation}
        So, we conclude that \eqref{eq:def-submodularity} holds for all $\setS,\frakS \subseteq 2^{\setN}$, and therefore, the reward $f$ is submodular.
\end{proof}
\subsection{Proof of Lemma \ref{lemma:subgaussian-vmag} (sub-Gaussianity of nodal voltages)}
\label{apdx:proof:subgaussian-vmag}
\begin{proof}
    Assumption \ref{assum:reporting-scenarios} implies that $p$ is a vector of sub-Gaussian variables with parameter $\frac{1}{2} \Delta$ \cite[Ex. 2.4]{wainwright_2019}. %and  \eqref{eq:l1:bounded-x-prob} follows by Markov's inequality and applying Definition \ref{def:sub-gaussian}, see \cite[Def. 2.2]{wainwright_2019} and \cite[Thm. 5.3]{lattimore_bandit_2020}. 
    For each node $i \in \setN$ and any $s \in \R$, the moment-generating function of fluctuations in $v_i$ %$v_i - \Expec[t]{v_i} = w_i^\T \Lambda^{-1} W^\T p = w_i^\T \Lambda^{-1} \rho$ 
    is conditionally bounded as
    \begin{subequations}
    \label{eq:v_i-mgf-bound}
        \begin{align}
            \Expec[]{e^{s\p{v_i - \Expec[]{v_i}}}} %&= \Expec[t]{e^{s \p{w_i^\T \Lambda^{-1} \rho} }}\\
            &= \Expec[]{\prod_{j=1}^n e^{s  W_{ij}\lambda_j^{-1}\rho_j }}\\
            &\overset{(1)}{=} \prod_{j=1}^n \Expec[]{e^{s W_{ij}\lambda_j^{-1} \rho_j }}\\
            &\overset{(2)}{\leq} \prod_{j=1}^n e^{\frac{1}{8} s^2  W_{ij}^2 \lambda_j^{-2} \Delta^2} \\
            &= e^{\frac{1}{8} s^2 \p{ \sum_{j=1}^n W_{ij}^2 \lambda_j^{-2} \Delta^2 }}.
        \end{align}
    \end{subequations}
    In the above display, step (1) is by the assumption of independence of the injections, and step (2) is by sub-Gaussianity of $\rho$ with parameter $\frac{1}{2} \Delta$. Thus, by Definition \ref{def:sub-gaussian} we see that $v_i$ is sub-Gaussian with parameter $\frac{1}{2}\Delta \norm{\Lambda^{-1} w_i}_2$, where $w_i$ is the $i$-th row of the eigenbasis of the Laplacian. 
    
    Next, we apply the Cram\'er-Chernoff bound \cite{wainwright_2019,lattimore_bandit_2020}. For any $s>0$, by Markov's inequality we obtain 
    \begin{equation}
    \label{eq:v_i-chernoff}
    \begin{split}
        \Prob[]{v_i -\Expec[t]{v_i} > \epsilon} &\leq \frac{  \Expec[]{e^{s \p{ v_i -\Expec[t]{v_i} } } }  }{s\epsilon}\\
        &\overset{\eqref{eq:v_i-mgf-bound}}{\leq} e^{ \frac{1}{8} s^2 \Delta^2 \norm{\Lambda^{-1} w_i}_2^2  -s \epsilon }.
    \end{split}    
    \end{equation}
    To make the upper bound \eqref{eq:v_i-chernoff} as small as possible, we minimize the exponent with respect to $s>0$, which yields 
    \[
    \inf_{s>0} \ \frac{1}{8} s^2 \Delta^2 \norm{\Lambda^{-1} w_i}_2^2  -s \epsilon = \frac{-2\epsilon^2}{\Delta^2 \norm{\Lambda^{-1} w_i}_2^2}.
    \]
    By including both tails in the bound, for all $i \in \setN$ we have 
    \begin{equation}
    \label{eq:concentration-vmag}
        \Prob[]{\abs{v_i - \Expec[]{v_i}} > \epsilon} \leq 2 \exp\c{ - \frac{2 \epsilon^2}{\Delta^2 \norm{\Lambda^{-1} w_i}_2^2 }}.
    \end{equation}
\end{proof}
%
    % \begin{remark}
    %     If the injection bounds $\Delta_j $ vary for each $j \in \setN$, the same steps bound on the moment-generating function \eqref{eq:v_i-mgf-bound} holds with parameter $\sigma_{v_i} := \frac{1}{2}\norm{\Lambda^{-1}\diag\p{\s{\Delta}_{j \in \setN}} w_i}$. %where $\oslash$ denotes elementwise division, $\lambda \in \R^n$ are the eigenvalues of the Laplacian, and $\Delta = [\Delta_j]_{j \in \setN}$ are the bound widths.
    % \end{remark}
%
\vspace{-1em}
\subsection{Proof of Theorem \ref{thm:max-v-concentration}}
\label{apdx:proof:max-x-concentration}
We emphasize that the following proof does not rely on independence of the voltage magnitudes. 

\begin{proof}
    To simplify notation, define the centered nodal voltages %$\sigma_{\sharp} :=  \frac{1}{2}\Delta \max_{i \in \setS} \norm{\Lambda^{-1} w_i}_2$, 
    $\tilde{v}_i := v_i - \Expec[]{v_i}$.
    %and $ \setS \wedge \abs{\tilde{v}_i} := \max_{i \in \setS} \abs{\tilde{v}_i}$. 
    Applying the union bound yields
    \begin{subequations}
    \label{eq:max-voltage-tail-bound}
        \begin{align}
           &\Prob[]{\max_{i \in \setS} \abs{\tilde{v}_i  } > \epsilon } = \Prob[]{\bigcup_{i \in \setS} \abs{\tilde{v}_i} > \epsilon }\\
            &\qquad \leq \sum_{i \in \setS} \Prob[]{ \abs{\tilde{v}_i} > \epsilon }\\
            &\qquad \overset{\eqref{eq:concentration-vmag}}{\leq} 2b \exp\c{ \frac{-2 \epsilon^2}{\Delta^2 \max\limits_{i \in \setS} \norm{\Lambda^{-1} w_i}_2^2 }  }.
        \end{align}
    \end{subequations}
    Therefore, in expectation, we have
    \begin{subequations}
        \label{eq:max-voltage-expectation-bound}
        \begin{align}
            \Expec[]{\max_{i \in \setS} \abs{\tilde{v}_i}} 
            &:= \int_{\epsilon=0}^\infty\Prob[]{\max_{i \in \setS} \abs{\tilde{v}_i  } \geq \epsilon} \mathsf{d\epsilon} \\
            &\hspace{-3.25em}\overset{(1)}{\leq} c + \int_{\epsilon=c}^\infty \sum_{i \in \setS} \Prob[]{\abs{\tilde{v}_i} > \epsilon }\mathsf{d\epsilon}\\
            &\hspace{-3.25em}\overset{(2)}{\leq} c + \int_{\epsilon=c}^\infty 2b \exp\c{ \frac{-2\epsilon^2}{\Delta^2 \max\limits_{i \in \setS} \norm{\Lambda^{-1} w_i}_2^2}}\mathsf{d\epsilon},
        \end{align}
    \end{subequations}
    where step (1) follows from the union bound and the fact that the probability of any event is upper bounded by $1$, combined with the fact that $\int_0^c \Prob[]{\max_{i \in \setS} \abs{\tilde{v}_i} \geq \epsilon}\mathsf{d\epsilon} \leq c$, and step (2) follows from the tail bound \eqref{eq:max-voltage-tail-bound}. Differentiating the upper bound with Leibniz's rule and solving for the minimizing $c$ yields the desired result \eqref{eq:l1:bounded-x-expec} with a factor of $\sqrt{\log(2b)}$, which can be sharpened to $\sqrt{\log(b)}$, as discussed in \cite[2.5.10]{Vershynin_2018}.
    \end{proof}
    % \begin{equation}
    %  \Expec[]{\max_{i \in \setS} \abs{\tilde{x}_i  }} \lesssim \sigma_\setS \sqrt{\log(2b)}.   
    % \end{equation}
    % Setting $\tilde{x}_i := x_i - \Expec[]{x_i}$, we obtain that
    % \begin{subequations}
    %     \begin{align}
    %         \Expec[]{\max_{i \in \setS} \abs{\tilde{x}_i  }} &= \int_{\epsilon=0}^\infty\Prob[]{\max_{i \in \setS} \abs{\tilde{x}_i} \geq \epsilon} \mathsf{d\epsilon} \\
    %         % &= \Prob[]{\bigcup_{i \in \setS} \abs{\tilde{x}_i} \geq \epsilon } \\
    %         &\overset{(1)}{\leq} c + \int_{\epsilon=c}^\infty \Prob[]{\max_{i \in \setS} \abs{\tilde{x}_i} \geq \epsilon}\mathsf{d\epsilon}\\
    %         &\overset{(2)}{\leq} c + \int_{\epsilon=c}^\infty 2d \exp\left\{ -\frac{2\epsilon^2}{\Delta^2}\right\}\mathsf{d\epsilon}
    %     \end{align}
    % \end{subequations}
   % 

% \begin{remark}
%     With a more complicated argument, it is known that can be sharpened to $\Expec[]{\max_{i \in \setS} \abs{\tilde{v}_i}} \lesssim   \sigma_{\setS} \sqrt{\log(b)}$
% \end{remark}

\vspace{-1em}
\subsection{Proof sketch for Theorem \ref{thm:regret}}

The following proof for the confidence ellipsoid is a straightforward update to \cite[Lemma 20]{kocak_spectral_2020}. 

\begin{proof}
First, we show that there exists a $C >0$ such that $\norm{\psi_\star}_{\Lambda} \leq C$. Let $v_t := v^\bullet + W \psi_* + \eta_t$, where $\eta_t := W\Lambda^{-1} W^\T p_t$ and $\psi_\star := \Lambda^{-1}W^\T p_\star$. Observe that
\begin{align*}
\norm{\psi^*}_{\Lambda} &= \sqrt{\tilde{v}^\T W \Lambda W^\T \tilde{v}}\\
&\overset{(1)}{=}\sqrt{\sum_{k=1}^n \lambda_k^{-1} \abs{\ip{p_\star}{w_k}}^2}\\
&\overset{(2)}{\leq} \norm{p}_2\sqrt{\sum_{k=1}^n \lambda_k^{-1}}\\
&\overset{(3)}{\leq} \norm{\Delta}_2\sqrt{\trace{R+XK}}:=C >0,
%\sqrt{v^\T W \Lambda W^\T v} = \sqrt{v^\T L v } = \sqrt{v^\T p}\\
% &\leq \sqrt{v^\T \max\c{\abs{\pmax},\abs{\pmin}}} := C.
\end{align*}
where step (1) is by definition, $v:=W\Lambda^{-1}W^\T p$, step (2) is by Cauchy-Schwarz and the fact that $\norm{w_k}_2 = 1$ for all $k$, step (3) is by the boundedness of $p$ and the definition of $\trace{\cdot}$, and $\lambda_k$ is the $k$-th eigenvalue of $L$.

Furthermore, note that $p_t \in [\pmin_t,\pmax_t]^n$ a.s., then by Lemma~\ref{lemma:subgaussian-vmag} we obtain that for all $i=1,\ldots,n$ the noise $(\eta_t)_i$ is sub-Gaussian with parameter at most
\begin{align*}
   \sigma_i &\leq \frac{1}{2}\Delta \norm{\Lambda^{-1} w_i}_2\\
    &\overset{(1)}{\leq} \frac{1}{2}\Delta \max_{w : \norm{w}_2 \leq 1} \norm{\Lambda^{-1} w}_2\\
    &\overset{(2)}{=} \frac{1}{2}\Delta \norm{\Lambda^{-1}}_{\sf op}
    %= \frac{1}{2}\Delta \lambda_{\sf max}(R)
    \overset{(3)}{=} \frac{1}{2}\Delta \norm{R + XK}_{\sf op}:=\sigma_\star, 
\end{align*}
where step (1) is true because $\norm{w_i}_2 \leq 1$ for all nodes $i$, step (2) is by definition of the operator norm, and step (3) is true because $R + XK \succeq 0$ by Assumption \ref{assum:known-phase}, hence, $\norm{R+XK}_{\sf op} := \sigma_{\sf max}(L^\dagger) = (\sigma_{\sf min}(L))^{-1} = \norm{\Lambda^{-1}}_{\sf op} $, where $\sigma_{\sf max}(\cdot)$ (resp. $\sigma_{\sf min}(\cdot)$) denotes the maximum (resp. minimum) singular value.

As $b=1$, redefine $V_t := \Lambda + \sum_{s=1}^{t-1} w_s w_s^\T$ and let $\xi_t := \sum_{\tau=1}^{t-1} \eta_\tau w_\tau$, where the $\eta_\tau$'s are independent sub-Gaussian random variables with parameter $\sigma_\star$. We have that for any node $i \in \setN$,  the voltage prediction error at time $t$ satisfies 
\begin{subequations}
    \begin{align}
        &\abs{\ip{w}{\hat{\psi} - \psi^*}} =\abs{\ip{w}{ -V_t^{-1} \Lambda^{-1} \psi^* + V_t^{-1} \xi_t }}\\
        &\overset{(1)}{\leq} \abs{\ip{V_t^{-\frac{1}{2}} w}{V_t^{-\frac{1}{2}} \Lambda \psi^* }} + \abs{\ip{V_t^{-\frac{1}{2}} w}{V_t^{-\frac{1}{2}} \xi_t}}\\
        &\overset{(2)}{\leq} \norm{w}_{V_t^{-1}} \p{\norm{\xi_t}_{V_t^{-1}} + \norm{\Lambda^{-1}\psi^*}_{V_t^{-1}}},
        %\\
        % &\overset{(3)}{\leq} \norm{w}_{V_t^{-1}}\p{ \norm{\xi_t}_{V_t^{-1}} + C }
    \end{align}
\end{subequations}
where steps (1) and (2) are by the triangle inequality and Cauchy-Schwarz, respectively.

By Assumption \ref{assum:effective-resistance-bound} and Lemma \ref{lemma:psd-L}, the initial gram matrix is positive definite, $V_0 := \lambda I + \Lambda \succ 0$. We need the following result, which appears in  \cite[Lemma 11]{abbasi-yadkori_improved_2011}, \cite[Lemma 19]{kocak_spectral_2020}, and \cite[Lemma 19.4]{lattimore_bandit_2020}.
\begin{lemma}
    \label{lemma:max-1-action}
    Let $\c{w_t}_{t=1}^m$ be a sequence of rows of $W$ chosen for the sampling strategy. Then, %Define $V_t = V_0 + \sum_{s \leq t} w_s w_s^\T$ to be the Gram matrix computed from the samples. Then,
    \begin{equation}
        \label{eq:lemma-max-1-action}
        \sum_{t=1}^m \min\c{1, \norm{w_t}^2_{V_{t-1}^{-1}}} \leq 2\log\p{\frac{\det V_m}{\det V_0}}. %2d\log\p{ \frac{\trace{V_0} + m}{n \det(V_0)^{1/n}}}.
    \end{equation}
\end{lemma}

We then apply \cite[Lemma 9]{abbasi-yadkori_improved_2011} under our Assumption \ref{assum:reporting-scenarios}; with probability at least $1-\delta$,
\begin{equation}
    \norm{\xi}_{V_{t}^{-1}} \leq \frac{1}{2} \Delta^2 \norm{\Lambda^{-1} w_t}_2^2 \log \p{\frac{\det V_t^{1/2}}{\delta \det \Lambda^{1/2} }}.
\end{equation}
By \cite[Thm. 2]{abbasi-yadkori_improved_2011}, for any $w \in \R^n$ and for all $t\geq 1$, as $\norm{\psi^*}_\Lambda \leq C$, then with probability at least $1-\delta$ it holds that
\begin{equation}
    \abs{\ip{w}{\hat \psi_t - \psi^*}} \leq \norm{w}_{V_t^{-1}} \underbrace{\p{ \sigma_\star \sqrt{ 2 \log \frac{\det V_t^{1/2 } }{ \delta \det\Lambda^{1/2} } } +C }}_{:=c_t},
\end{equation}
where $\sigma_\star,C$ are as defined in the bounds above.
% Let $w_\star := \argmax_{\c{w_i}_{i=1}^n} \ip{w_i}{\psi_\star}$ denote the clarivoyant optimal sensing strategy and let $\bar{r}(t)$ denote instantaneous regret at time $t$. With probability at least $1-\delta$, for all $t$, we obtain
% \begin{align*}
%     \bar{r}(t)&= \ip{w_\star}{\psi_\star} - \ip{w_{\pi(t)}}{\psi_\star}
% \end{align*}
Then, the regret bound for $b=1$ is found by directly invoking \cite[Thm. 1]{valko_spectral_bandits_2014} and Lemma~\ref{lemma:max-1-action}. For $b \neq 1$, a $\tilde{\setO}(d\sqrt{bm})$ bound can be produced by combining \cite[Thm. 1]{yue_linear_2011} and \cite[Thm. 1]{valko_spectral_bandits_2014}, then applying Lemma~\ref{lemma:reward-properties}.
\end{proof}

% \subsection{\textcolor{red}{\textbf{Potentially to be removed}: Pedro's algorithm}}

% \begin{equation}
%     \min_{R} \trace{H^\T (\Sigma_Y Q)^{-1} H}^{-1} \ \st R_i \geq 0,\ \sum_{i=1}^m R_i = R_{\sf max}
% \end{equation}
% Include adaptation of weighting quantities according to how close they are to limits.

% - distortion is squared quantization error.

% \textbf{Rate distortion function:}

% \[
% R(D) = \begin{cases}
%     \frac{1}{2} \log(\sigma^2/D) & 0 \leq D \leq \sigma^2\\
%     0 & D > \sigma^2
% \end{cases}
% \]
% \[
% D(R) = \min(\sigma^2,\sigma^2 2^{-2R})
% \]
% What we want to return:
% \textbf{Voltage measurements with this much accuracy and this bitrate}
% Important questions:
% - Which ones do we need to know accurately
% - At what bitrate
% - How much information do we want each meter

\end{document}